\documentclass[letterpaper, 10pt, conference]{ieeeconf}
\IEEEoverridecommandlockouts  
\usepackage[utf8]{inputenc}
\usepackage[english]{babel}
\usepackage{lscape}
\usepackage{longtable}
\usepackage{booktabs}
\usepackage{graphicx}
\usepackage{multirow}
\usepackage{setspace}
\usepackage{hyperref}
\usepackage{amsmath,amssymb,amsfonts}
\usepackage{cite}
\usepackage{color}

\usepackage[margin=1in]{geometry}
\usepackage[caption=false, font=footnotesize]{subfig}
\hypersetup{
    colorlinks=true,
    linkcolor=blue,
    filecolor=magenta,      
    urlcolor=blue,
    citecolor=magenta
}

\newtheorem{thm}{Theorem}[section]

\newtheorem{asm}{Assumption}[section]

\newtheorem{rem}{Remark}

\title{Generalized Asynchronous Event-Triggered Measurement and Control for Non-Linear Systems}
\author{Daniel A. Williams$^{1}$, Airlie Chapman$^{2}$, Chris Manzie$^{1}$  %<-this % stops a space
\thanks{$^{1}$Daniel A. Williams and Chris Manzie are with the Department of Electrical Engineering, The University of Melbourne, 3052 Parkville, Australia.}%
\thanks{$^{2}$Airlie Chapman is with the Department of Mechatronic Engineering,
        The University of Melbourne, 3052 Parkville, Australia.}%
}
\date{}

\begin{document}
\maketitle
\begin{abstract}                % Abstract of not more than 250 words.
With the increasing ubiquity of networked control systems, various strategies for sampling constituent subsystems' outputs have emerged. 
In contrast with periodic sampling, event-triggered control provides a way to efficiently sample a subsystem and conserve network resource usage, by triggering an update only when a state-dependent error threshold is satisfied.
Herein we describe a novel scheme for asynchronous event-triggered measurement and control (ETC) of a nonlinear plant using sampler subsystems with hybrid dynamics.
We extend existing ETC literature by adopting a more general representation of the sampler subsystem dynamics that do not require trigger periodicity or simultaneity,
thus accommodating different sampling schemes for both synchronous and asynchronous ETC applications.
We ensure that the plant and controller trigger rules are not susceptible to Zeno behavior by employing auxiliary timer variables in conjunction with state-dependent error thresholds.
We conclude with a numerical example in order to illustrate important practical considerations when applying such schemes.
\end{abstract}

\section{Introduction}

\begin{figure*}[t]
    \centering
    \includegraphics[width=1.45\columnwidth]{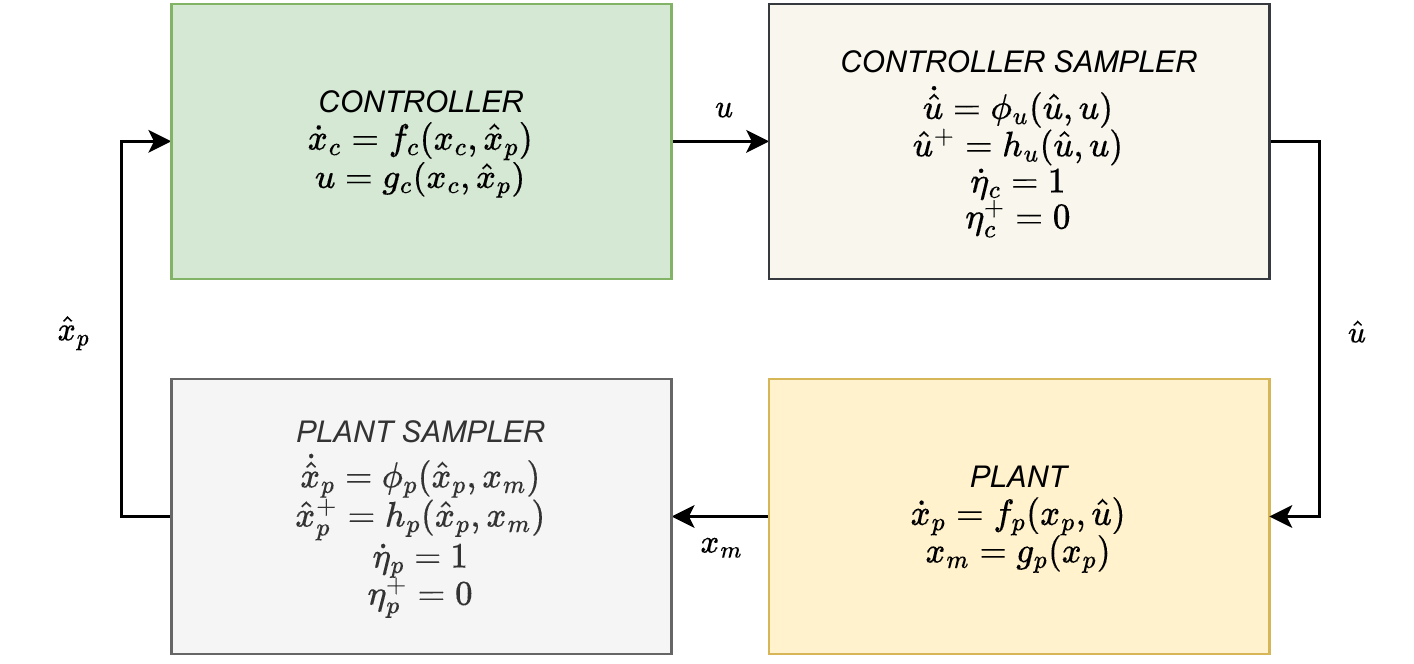}
    \caption{Overview of Proposed Event-Triggered Control Scheme.}
    \label{fig:framework_noidj}
\end{figure*}

The advent of fast and reliable communication channels has permitted the ongoing development of networked control systems.
A central problem in networked control surrounds the scheduling of sample updates for both the plant and controller subsystems \cite{tabuada_event-triggered_2007}.
Traditional fixed sampling rate approaches require the rate to be sufficiently high in order to minimize the effect of sampling errors, however this can generate redundant samples or incur transmission costs during periods of reasonably steady outputs from the plant or controller.
Alternatively, event-triggered transmission seeks a trade-off between sampling error and transmission frequency by triggering samples only when the sampling error exceeds a (potentially state-dependent) threshold (\cite{astrom_event_2008}).

Determining appropriate conditions for sample triggering is non-trivial, however a growing body of literature has emerged on the topic (cf. \cite{tabuada_event-triggered_2007,heemels_introduction_2012,postoyan_framework_2015,sanfelice_hybrid_2020,maass_state_2021}).
In \cite{tabuada_event-triggered_2007} the author introduces a state-dependent error threshold for triggering samples of continuous-time control systems, and conducts detailed analysis of such schemes when applied to linear and nonlinear control systems. 
Building on this, \cite{postoyan_framework_2015} generalizes event-triggered control by devising a hybrid system formulation for a non-linear plant and controller accommodating the use of auxiliary variables in sampler triggering.
The authors then derive sufficient conditions for guaranteeing the existence of a minimum inter-trigger interval and global (pre)-asymptotic stability of the origin for different applications of the ETC framework.
A limitation of both works however is an implicit requirement for simultaneous updates of both the plant sampler and controller sampler.

To overcome this we may instead consider asynchronous ETC, which does not require simultaneous state sampling \cite{mazo_asynchronous_2014,fu_periodic_2016,li_asynchronous_2021}.
Such schemes are relevant for scenarios wherein the plant and controller are not co-located, 
% leading to intermittent communication between the two entities 
particularly when a human supervises an autonomous robotic system, e.g. for swarm foraging \cite{nam_predicting_2017}
% , foraging \cite{mahani_bayesian_2020}, 
and search and rescue \cite{drew_multi-agent_2021}.
% Focusing on asynchronous ETC for nonlinear plants and controllers, we identify three desirable features: uniform global asymptotic stability, no triggers exhibiting Zeno behavior, and the ability to use sampling methods other than the zero-order hold (ZOH).
% Some of these features appear in existing works. 
In \cite{abdelrahim_robust_2017} the authors propose an asynchronous ETC scheme that prevents Zeno behavior by enforcing a minimum time between consecutive triggers (a strategy also followed in \cite{dolk_event-triggered_2017,scheres_distributed_2022}), however only input-to-state stability and $\mathcal{L}_p$ stability are demonstrated, and non-ZOH sampling is not considered.
In \cite{scheres_distributed_2022} the authors present conditions for ensuring dissipativity of the closed-loop system, while in \cite{dolk_event-triggered_2017} the authors derive a stronger result involving global asymptotic stability, however these works again rely on ZOH sampling.
Conversely, while \cite{wang_unifying_2020} present an ETC scheme that permits non-ZOH sampling methods and demonstrates uniform global asymptotic stability of the proposed system, the scheme is synchronous.

To address these gaps, we propose a new scheme for asynchronous ETC that permits non-ZOH sampling, exhibits uniform global asymptotic stability, and prevents the occurrence of Zeno behavior in triggers.
While we use similar notation and reasoning to \cite{wang_unifying_2020} in order to develop our scheme, we decompose the closed-loop system spatially into plant-related and controller-related subsystems as opposed to aggregating states by type (i.e. error and non-error variables).
This therefore allows consideration of asynchronous triggering schemes based only on information locally available to each sampler, and results in a stronger stability result than those in \cite{dolk_event-triggered_2017} and \cite{scheres_distributed_2022}.

\section{Preliminaries}\label{sec:prels}

Define $\mathbb{R}_{\geq 0}$ as the set of non-negative real numbers and $\mathbb{Z}_{\geq 0}$ as the set of non-negative integers.
We represent the identity function as $\mathbb{I}$.
A function $\beta:\mathbb{R}_{\geq 0}\rightarrow\mathbb{R}_{\geq 0}$ is of class $\mathcal{K}_{\infty}$ if it is a continuous function, zero at the origin, strictly increasing and unbounded.
We denote the distance of a vector $x\in\mathbb{R}^n$ from a set $\mathcal{A}\subset\mathbb{R}_{\geq 0}$ as $|x|_\mathcal{A}=\inf\{|x-y|:y\in A \},$ where $|\cdot|$ is the Euclidean norm.
A function $h:\mathbb{R}^n\rightarrow\mathbb{R}^m$ is proper if $|h(x)|\rightarrow\infty$ as $|x|\rightarrow\infty$.

Given a function $f:\mathbb{R}^n\rightarrow\mathbb{R}^m$ and a set $\mathcal{M}\subset\mathbb{R}^m$, let $f^{-1}(\mathcal{M})={x:f(x)\in\mathcal{M}}$.
When $m=1$, $\mathcal{M}\in\mathbb{R}$ and we define $f^{-1}_{\leq}(\mathcal{M})=\{x:f(x)\leq M\}$.

The tangent cone $T_S(x)$ to a set $S\subset\mathbb{R}^n$ at a point $x\in\mathbb{R}^n$ is the set 
$\{w\in\mathbb{R}^n: w = \lim_{i\rightarrow\infty}\frac{x_i-x}{\tau_i}\}$
where $x_i\in S$ is such that $
\lim_{i\rightarrow\infty}x_i = x$ and $\tau_i>0$ is such that $\lim_{i\rightarrow\infty}\tau_i = 0$.

The generalized Clarke directional derivative of a locally Lipschitz function $U:\mathbb{R}^n\rightarrow\mathbb{R}$ 
at $x\in\mathbb{R}^n$ in the direction of $v\in\mathbb{R}^n$ is given by 
$$U^{\circ}(x;v) = \limsup_{h\rightarrow 0^+,y\rightarrow x}(U(y+hv)-U(y))/h.$$

\section{Problem Formulation}\label{sec:probform}

We consider a nonlinear plant and a dislocated controller with dynamic samplers at each output.
We first provide an overview of each subsystem, then define flow and jump conditions for closed-loop operation.

\subsection{Subsystems}

We employ four subsystems as depicted in Figure \ref{fig:framework_noidj}: a plant, a plant sampler, the event-triggered controller, and the controller sampler.

\subsubsection{Plant} 
Given a control input $\hat{u}\in\mathbb{R}^{n_{\hat{u}}}$, the plant state $x_p\in\mathbb{R}^{n_p}$ evolves according to
\begin{align}
    \dot{x}_p = f_p(x_p,\hat{u}),\label{eq:plantstate}
\end{align}
where the flow map $f_p:\mathbb{R}^{n_p}\times\mathbb{R}^{n_u}\rightarrow\mathbb{R}^{n_p}$ is continuous.
The plant then generates a measured state
\begin{align}
    x_m &= g_p(x_p) \in \mathbb{R}^{n_m}, \label{eq:measuredplantstate}
\end{align} 
using the static map $g_p:\mathbb{R}^{n_p}\rightarrow\mathbb{R}^{n_m}$, $n_m\leq n_p$.
To ensure that the plant sampler can recover the plant state from the measured state, the plant must be observable; we employ the following characterization of observability for nonlinear systems from \cite{hermann_nonlinear_1977}:

\begin{asm}\label{asm:oability}
For the system given by \eqref{eq:plantstate}--\eqref{eq:measuredplantstate}, there exists an equivalence relation $I$ on $\mathbb{R}^{n_p}$ such that for all $x_p\in\mathbb{R}^{n_p}$, $I(x_p)=\{x_p\}$.
\end{asm}

\subsubsection{Plant Sampler} 
We consider a plant sampler that conveys information about the plant state $x_p$ to the controller.
The first state of the plant sampler generates an estimated plant state $\hat{x}_p\in\mathbb{R}^{n_p}$ using the measured state $x_m$ and is governed by the hybrid dynamics
\begin{align}
    \dot{\hat{x}}_p&=\phi_p(\hat{x}_p,x_m), \label{eq:plantsamplerflow}\\
    \hat{x}_p^+&=h_p(\hat{x}_p,x_m),\label{eq:plantsamplerjump}
\end{align}
where the flow map $\phi_p:\mathbb{R}^{n_p}\times\mathbb{R}^{n_{m}}\times\mathbb{R}^{n_{\hat{u}}}\rightarrow\mathbb{R}^{n_{p}}$ and the jump map $h_p:\mathbb{R}^{n_p}\times\mathbb{R}^{n_m}\times\mathbb{R}^{n_{\hat{u}}}\rightarrow\mathbb{R}^{n_{p}}$ are continuous functions.
The second state of the plant sampler, $\eta_p\in\mathbb{R}_{\geq0}$, is an auxiliary variable that records the time elapsed since the last local sampling event, with dynamics given by
\begin{align}
    \dot{\eta}_p &= 1, \label{eq:planttimerflow}\\
    \eta_p^+ &= 0.\label{eq:planttimerjump}
\end{align}
We use $\eta_p$ to impose a minimum plant sampling interval $\tau_p>0$ and a maximum plant sampling interval $\tau_{\pi}>\tau_p$.
Local sampling events cause both $\hat{x}_p$ and $\eta_p$ to jump, while between sampling events $\hat{x}_p$ and $\eta_p$ flow.
The plant sampling error is denoted by $e_p = \hat{x}_p - x_p\in\mathbb{R}^{n_p}$. 

Defining an aggregate plant state $q_p = [x_p^T, e_p^T, \eta_p^T]^T\in\mathbb{R}^{n_{q,p}}$, we select a locally Lipschitz function $V_p:\mathbb{R}^{n_{q,p}}\rightarrow\mathbb{R}_{\geq0}$ and a continuously differentiable function $W_p:\mathbb{R}^{n_p}\rightarrow\mathbb{R}_{\geq0}$ so that sampling is triggered when either $\eta_p=\tau_{\pi}$ or both $V_p(q_p)\leq W_p(e_p)$ and $\eta_p\geq\tau_p$, with the latter condition preventing Zeno behavior (cf. \cite[Sec. V-C]{postoyan_framework_2015}).
To this end we define the local flow and jump sets
\begin{align}
    C_p &:= \{q_p: V_p(q_p)> W_p(e_p)\&\eta_p<\tau_{\pi}\parallel\eta_p<\tau_p\},\label{eq:plocalflowset}\\
    D_p &:= \{q_p: V_p(q_p)\leq W_p(e_p)\&\eta_p\geq\tau_p\parallel\eta_p=\tau_{\pi}\}.\label{eq:plocaljumpset}
\end{align}

\subsubsection{Controller}
The controller monitors the sampled plant state $\hat{x}_p$ and updates its output $u\in\mathbb{R}^{n_u}$ dynamically with
\begin{align}
    \dot{x}_c &= f_c(x_c,\hat{x}_p),\\
    u &= g_c(x_c,\hat{x}_p),
\end{align}
where the flow map $f_c:\mathbb{R}^{n_c}\times\mathbb{R}^{n_p}\rightarrow\mathbb{R}^{n_{c}}$ and $g_c:\mathbb{R}^{n_c}\times\mathbb{R}^{n_p}\rightarrow\mathbb{R}^{n_{u}}$ is continuous.

\subsubsection{Controller Sampler}
Similar to the plant sampler, the controller sampler's first state $\hat{u}\in\mathbb{R}^{n_u}$ is governed by the hybrid dynamics
\begin{align}
    \dot{\hat{u}}&=\phi_u(\hat{u},u), \\
    \hat{u}^+&=h_u(\hat{u},u),
\end{align}
where the flow map $\phi_u:\mathbb{R}^{n_u}\times\mathbb{R}^{n_{u}}\rightarrow\mathbb{R}^{n_u}$ and the jump map $h_u:\mathbb{R}^{n_u}\times\mathbb{R}^{n_{u}}\rightarrow\mathbb{R}^{n_u}$ are continuous functions representing the chosen sampling method.
The second state $\eta_c\in\mathbb{R}_{\geq0}$ records time elapsed since the last local sampling event with the hybrid dynamics
\begin{align}
    \dot{\eta}_c &=1, \label{eq:controllersamplertimerflow} \\
    \eta_c^+ &= 0. \label{eq:controllersamplertimerjump} 
\end{align}
We use $\eta_c$ to impose a minimum controller sampling interval $\tau_c > 0$ and a maximum sampling interval $\tau_{\kappa}>\tau_c$.
Local sampling events cause both $\hat{u}$ and $\eta_c$ to jump, otherwise $\hat{u}$ and $\eta_c$ flow. 
The controller sampling error is denoted by $e_u = \hat{u} - u\in\mathbb{R}^{n_u}$.

Defining the aggregate controller state $q_c = [x_c^T, e_u^T, \eta_c^T]^T\in\mathbb{R}^{n_{q,c}}$, we choose a locally Lipschitz function $V_c:\mathbb{R}^{n_{q,c}}\times\mathbb{R}_{\geq0}\rightarrow\mathbb{R}_{\geq0}$ and a continuously differentiable function $W_u:\mathbb{R}^{n_u}\rightarrow\mathbb{R}_{\geq0}$ so that sampling is triggered when either $\eta_c=\tau_{\kappa}$ or both $V_c(q_c)\leq W_u(e_u)$ and $\eta_c\geq\tau_c$ (again to prevent Zeno behavior).
We thus define the local flow and jump sets
\begin{align}
    C_c &:= \{q_c: V_c(q_c)> W_u(e_u)\&\eta_c<\tau_{\kappa}\parallel\eta_c<\tau_c\},\label{eq:clocalflowset}\\
    D_c &:= \{q_c: V_c(q_c)\leq W_u(e_u)\&\eta_c\geq\tau_c\parallel \eta_c=\tau_{\kappa}\}.\label{eq:clocaljumpset}
\end{align}

\subsection{Closed-Loop Dynamics}

We next consider a consolidation of the four subsystems (plant, plant sampler, controller, controller sampler) into two subsystems (plant and its sampler, controller and its sampler).
For the networked system we define the flow set 
\begin{align} C&=\{(q_p,q_c)\in\mathbb{R}^{n_p+n_c}:q_p\in C_p\,\&\,q_c\in C_c\},\end{align}
the jump set 
\begin{align} D&=\{(q_p,q_c)\in\mathbb{R}^{n_p+n_c}:q_p\in D_p\,\parallel\,q_c\in D_c\},\end{align}
the flow maps
\begin{align}
F_1(q_p,q_c)&=\begin{bmatrix}
&f_p(x_p,\hat{u})\\ 
&\phi_p(\hat{x}_p,x_m)-f_p(x_p,\hat{u})\\ 
&1\end{bmatrix}, \\
F_2(q_p,q_c)&=\begin{bmatrix}
    &f_c(x_c,\hat{x}_p)^T\\
    &\phi_u(\hat{u},u)^T-\frac{d g_c}{d t}(x_c,\hat{x}_p)\\
    &1
\end{bmatrix},
\end{align}
wherein we substitute $x_m=g_p(x_p)$, $\hat{x}_p = x_p + e_p$, and $\hat{u} = u + e_u = g_c(x_c,x_p + e_p) + e_u$,
and the jump maps 
\begin{align}
    G_1(q_p,q_c) &= \begin{cases}
    \begin{bmatrix}
        x_p\\0\\0
    \end{bmatrix}, &(q_p,q_c)\in D_p\times C_c,\\
    \begin{bmatrix}
        x_p\\e_p\\ \eta_p
    \end{bmatrix}, &(q_p,q_c)\in C_p\times D_c,\\
    \left\{\begin{bmatrix}
        x_p\\0\\0
    \end{bmatrix},\begin{bmatrix}
        x_p\\e_p\\ \eta_p
    \end{bmatrix}\right\}, &(q_p,q_c)\in D_p\times D_c, \label{eq:pdoublejump}
    \end{cases}
\end{align}
and
\begin{align}
    G_2(q_p,q_c) &= \begin{cases}
    \begin{bmatrix}
        x_c\\ e_u\\ \eta_c
    \end{bmatrix}, &(q_p,q_c)\in D_p\times C_c,\\
    \begin{bmatrix}
        x_c\\0\\0
    \end{bmatrix}, &(q_p,q_c)\in C_p\times D_c,\\
    \left\{\begin{bmatrix}
        x_c\\ e_u\\ \eta_c
    \end{bmatrix},\begin{bmatrix}
        x_c\\0\\0
    \end{bmatrix}
    \right\}, &(q_p,q_c)\in D_p\times D_c. \label{eq:cdoublejump}
    \end{cases}
\end{align}
Note that \eqref{eq:pdoublejump} and \eqref{eq:cdoublejump} denote successive jumps due to both samplers triggering simultaneously (cf. \cite{abdelrahim_robust_2017}).
Similar to \cite{liberzon_lyapunov-based_2014} and \cite{wang_unifying_2020}, we thus represent the hybrid dynamics of $(q_p,q_c)$ as
\begin{align}
    \begin{bmatrix}
        \dot{q}_p\\
        \dot{q}_c
    \end{bmatrix} &\in \begin{bmatrix}
        F_1(q_p,q_c)\\
        F_2(q_p,q_c)
    \end{bmatrix}, &(q_p,q_c)\in C,\label{eq:sysflows} \\
    \begin{bmatrix}
        q^+_p\\
        q^+_c
    \end{bmatrix} &\in \begin{bmatrix}
        G_1(q_p,q_c)\\
        G_2(q_p,q_c)
    \end{bmatrix}, &(q_p,q_c)\in D. \label{eq:sysjumps}
\end{align}

We make the following assumption adapted from \cite[Assumption III.2]{liberzon_lyapunov-based_2014} in order to later apply the hybrid small-gain theorem from \cite{liberzon_lyapunov-based_2014}.
\begin{asm}\label{asm:smallgain}
The functions $V_p$ and $V_c$ used in \eqref{eq:plocalflowset}--\eqref{eq:plocaljumpset} and \eqref{eq:clocalflowset}--\eqref{eq:clocaljumpset} respectively satisfy the following:
\begin{enumerate}
    \item \label{asm:smallgain1} given the sets $\mathcal{A}_p=\{q_p:x_p = 0, e_p=0, \eta_p\in[0,\tau_{\pi}]\}$ and $\mathcal{A}_c=\{q_c:x_c = 0, e_u=0, \eta_c\in[0,\tau_{\kappa}]\}$ there exist functions $\underline{\alpha}_p,\bar{\alpha}_p,\underline{\alpha}_c,\bar{\alpha}_c\in\mathcal{K}_{\infty}$ such that for all $q_p\in\mathbb{R}^{n_{q,p}}$ and $q_c\in\mathbb{R}^{n_{q,c}}$,
    \begin{align}
        &\underline{\alpha}_p(|q_p|_{\mathcal{A}_p}) \leq V_p(q_p) \leq \bar{\alpha}_p(|q_p|_{\mathcal{A}_p}),\\
        &\underline{\alpha}_c(|q_c|_{\mathcal{A}_c}) \leq V_c(q_c) \leq \bar{\alpha}_c(|q_c|_{\mathcal{A}_c});
    \end{align} 
    \item \label{asm:smallgain2} there exist functions $\chi_p,\chi_c\in\mathcal{K}_{\infty}\cup\{0\}$, $\alpha_p,\alpha_c:\mathbb{R}_{\geq 0}\rightarrow \mathbb{R}_{\geq 0}$ such that for all $(q_p,q_c)\in C$,
    \begin{align}
        V_p(q_p)\geq \chi_p(V_c(q_c))\implies& \nonumber\\
        V_p^{\circ}(q_p,F_1(q_p,q_c))&\leq -\alpha_p(|q_p|_{\mathcal{A}_p}),\\
        V_c(q_c)\geq \chi_c(V_p(q_p))\implies& \nonumber\\
        V_c^{\circ}(q_c,F_2(q_p,q_c))&\leq -\alpha_c(|q_c|_{\mathcal{A}_c});
    \end{align}
    \item \label{asm:smallgain3} for all $(q_p,q_c)\in D$, 
    \begin{align}
        V_p(G_1(q_p,q_c))\leq V_p(q_p),\\
        V_c(G_2(q_p,q_c))\leq V_c(q_c);
    \end{align}
    \item \label{asm:smallgain4} $\chi_p \circ \chi_c (s) < s $ for all $s>0$;
    \item \label{asm:smallgain5} there exists a function $\rho\in\mathcal{K}_{\infty}$ that is continuously differentiable on $(0,\infty)$ satisfying both $\chi_p(r)<\rho(r)<\chi_c^{-1}(r)$ and $\rho'(r)>0$ for all $r>0$.
\end{enumerate}
\end{asm}

\begin{rem}\label{rem:hard}
In the context of \eqref{eq:sysflows}--\eqref{eq:sysjumps}, the conditions of Assumption \ref{asm:smallgain} are conservative and imply that growth in the error variables $e_p$ and $e_u$ depends on the choice of controller and plant gains, in addition to the minimum sampling periods $\tau_p$ and $\tau_c$.
For all but the most trivial systems, it is challenging to analytically determine the functions necessary to satisfy Assumption \ref{asm:smallgain}. 
Nevertheless, we may still use the intuition underlying these conditions to inform the choice of parameters for more complex systems (cf. Section \ref{sec:eg}).
\end{rem}

Finally, defining an aggregate state $q = [q_p^T, q_c^T]^T\in\mathbb{R}^{n_q}=\mathbb{R}^{n_{q,p}+n_{q,c}}$, we may represent \eqref{eq:sysflows}--\eqref{eq:sysjumps} as a well-posed hybrid dynamical system of the form
\begin{align}
    \dot{q} &\in F(q), &q\in C,\label{eq:generalFlow}\\
    q^+ &\in G(q), &q\in D,\label{eq:generalJump}
\end{align}
where given $q= [q_p^T, q_c^T]^T$, the flow map $F(q) = [F_1(q_p,q_c)^T, F_2(q_p,q_c)^T]^T$ and jump map $G(q) = [G_1(q_p,q_c)^T, G_2(q_p,q_c)^T]^T$.
\begin{asm}\label{asm:maxcomplete}
$G(D)\subset C \cup D$ and $F (q) \in T_C (q)$ for any $q \in C \backslash D$.
\end{asm}

\section{Main Result}\label{sec:mainres}

Having introduced a hybrid event-triggered measurement and control scheme in Section \ref{sec:probform}, we shall now focus on the solutions of the hybrid system \eqref{eq:generalFlow}--\eqref{eq:generalJump}. 
The following theorem concerns the stability of the set $\mathcal{A}_p\times\mathcal{A}_c$ 
for solutions of \eqref{eq:generalFlow}--\eqref{eq:generalJump}.

\begin{thm}
\label{thm:ugasA}
Given the system \eqref{eq:generalFlow}--\eqref{eq:generalJump} under Assumptions \ref{asm:oability} and \ref{asm:smallgain}, if the set $\mathcal{A}_p\times\mathcal{A}_c$ is compact, $G(D)\subset C\cup D$ and $F(q)\in T_C(q)$ for any $q\in C\\D$, then the set $\mathcal{A}_p\times\mathcal{A}_c$ is uniformly globally asymptotically stable. 
\end{thm}

\begin{proof}
Let $\phi$ denote a solution of the system \eqref{eq:sysflows}--\eqref{eq:sysjumps} under Assumptions \ref{asm:oability} and \ref{asm:smallgain}. 
Define a storage function $U(q) = \max(V_p(q_p),\rho(V_c(q_c)))$. 
By \cite[Thm. III.3]{liberzon_lyapunov-based_2014}, the following hold:
\begin{enumerate}
    \item there exist functions $\underline{\alpha}_U, \bar{\alpha}_U\in \mathcal{K}_{\infty}$ such that for all $q\in\mathbb{R}_{n_q}$,
    \begin{align}\underline{\alpha}_U(|q|_{\mathcal{A}})\leq U(q)\leq \bar{\alpha}_U(|q|_{\mathcal{A}}); \label{eq:libthm3.1}
    \end{align}
    \item  there exists a positive definite function $\alpha_U:\mathbb{R}_{\geq 0}\rightarrow \mathbb{R}_{\geq 0}$ such that for all $q\in C \backslash \mathcal{A}$,
    \begin{align}
        U^{\circ}(q,F(q))\leq -\alpha_U(|q|_{\mathcal{A}});\label{eq:libthm3.2}
    \end{align}
    \item For all $q\in D$, 
    \begin{align}
        U(G(q))\leq U(q). \label{eq:libthm3.3}
    \end{align}
\end{enumerate}

Due to the construction of the jump sets \eqref{eq:plocaljumpset} and \eqref{eq:clocaljumpset}, a continuous time interval of at least $\tau_p$ elapses between any pair of jumps in $q_p$; these jumps may be caused either by the plant sampler alone (which we denote \textit{`P'}) or by both samplers (which we denote \textit{`B'}). 
Similarly a continuous time interval of at least $\tau_c$ elapses between any pair of jumps in $q_c$ caused either by both samplers, or the controller sampler alone (which we denote \textit{`C'}). 
If any sequence of at least three jumps occurs, then at least two of these jumps involve the same state (e.g. the sequence \textit{`CPC'} involves two jumps in $q_c$, while the sequence \textit{`BP'} involves two jumps in $q_p$. 
Hence an interval of at least $\tau:=\min(\tau_p,\tau_c)$ must elapse between the first and the third jump.
Define the quantity $\tau_a = \frac{\tau}{2}$.
We recall from \cite{goebel_hybrid_2012} that the solution's hybrid time domain $\mathrm{dom}\phi$ satisfies
\begin{align}
    \mathrm{dom}\phi \cap ([0,t]\times\{0,...,j\}) &= \bigcup_{i\in\{0,...,j\}}[t_i,t_{i+1}]\times\{i\}\label{eq:solution domain}
\end{align}
for any point $(t,j)\in \mathrm{dom}\phi$ and sequence of times $0=t_0\leq t_1\leq...\leq t_{j+1}=t$.
Given $(s,i),(t,j)\in\mathrm{dom}\phi$, we wish to show that if $(s+i)\leq(t+j)$, then the inequality
\begin{align}
    j-i\leq\frac{t-s}{\tau_a}+1 \label{eq:quasidwelltime}
\end{align}
holds for all $(j-i)\in\mathbb{N}\cup\{0\}$.
If $j-i=0$, $(s+i)\leq(t+j)$ reduces to $s\leq t$; since $\tau_a>0$, \eqref{eq:quasidwelltime} holds.
If $j-i = 1$,  $(s+i)\leq(t+j)$ reduces to $s\leq t+1$, however by the definition of $\mathrm{dom}\phi$ necessarily $s\leq t$ for $i<j$, so \eqref{eq:quasidwelltime} holds.
If $j-i = 2n$, $n\in\mathbb{N}$, then we apply the result that $t-s \geq n\tau=2n\tau_a$ to show that \eqref{eq:quasidwelltime} holds.
If $j-i = 2n+1$, $n\in\mathbb{N}$, we may again apply $t-s \geq 2n\tau_a$ to show that \eqref{eq:quasidwelltime} holds for all $(j-i)\in\mathbb{N}\cup\{0\}$.
Given \eqref{eq:quasidwelltime} and taking $(s+i)=0$, $(t+j)\geq 0$ implies 
\begin{align}
    j\leq \frac{t}{\tau_a}+1. \label{eq:jumptimebound}
\end{align}
We may thus apply \eqref{eq:jumptimebound} to the inequality $(t+j)\geq T$ for some $T>0$, yielding the inequality
\begin{align}
    t\geq \gamma_r(T) - N_r, \label{eq:timebound}
\end{align}
where $\gamma_r(s) = (1+\tau_a^{-1})^{-1}s$ for $s>0$, and $N_r=(1+\tau_a^{-1})^{-1}$.
Given \eqref{eq:libthm3.1}, \eqref{eq:libthm3.2}, \eqref{eq:libthm3.3}, and \eqref{eq:timebound}, we may now apply \cite[Prop. 3.27]{goebel_hybrid_2012} to demonstrate that the set $\mathcal{A}_p\times\mathcal{A}_c$ is uniformly globally pre-asymptotically stable.
Since $\mathcal{A}_p\times\mathcal{A}_c$ is compact and following the approach of \cite{postoyan_framework_2015}, we use \cite[Prop. 6.10]{goebel_hybrid_2012} with Assumption \ref{asm:maxcomplete} to show that maximal solutions of \eqref{eq:generalFlow}--\eqref{eq:generalJump} are complete, hence $\mathcal{A}_p\times\mathcal{A}_c$ is uniformly globally asymptotically stable. 
% \qed
\end{proof}

\section{Numerical Example}\label{sec:eg}

\begin{figure}[t]
    \centering
    \includegraphics[width=0.9\columnwidth]{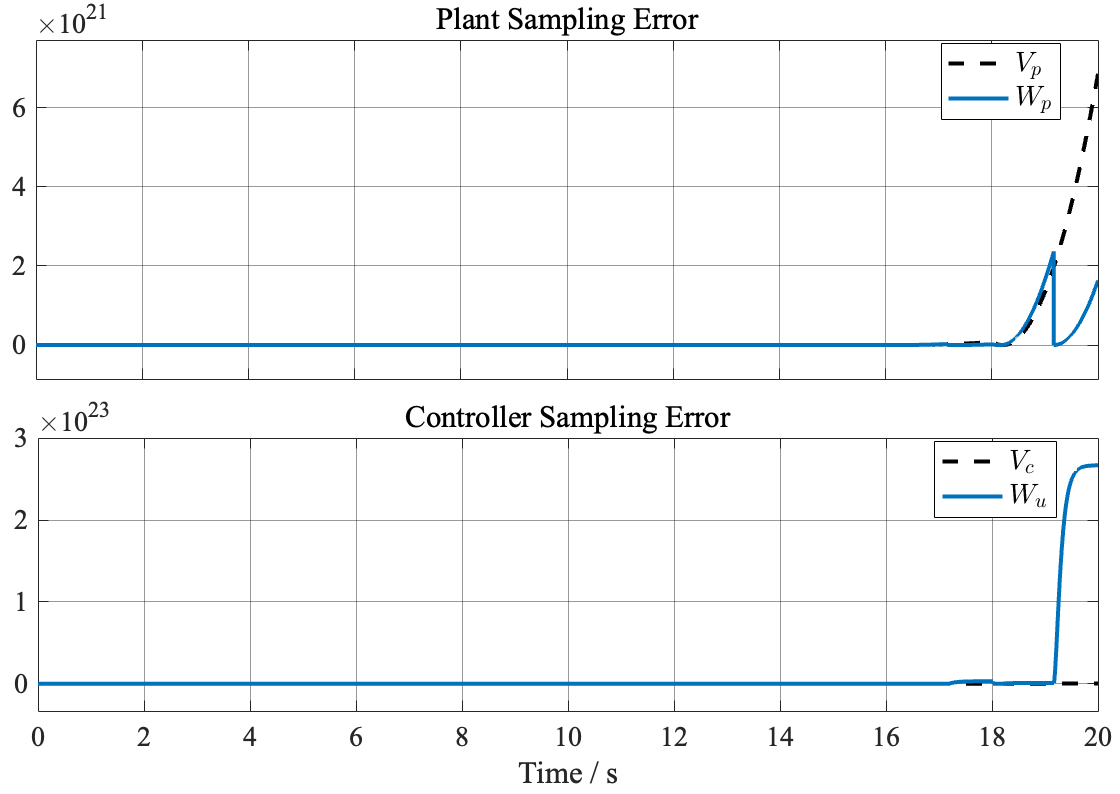}
    \caption{Sampling Error for $k_p=10$, $\tau_p = 1$ s, $\tau_c = 2$ s.}
    \label{fig:1-2-10 thres}
\end{figure}
\begin{figure}[t]
    \centering
    \includegraphics[width=0.9\columnwidth]{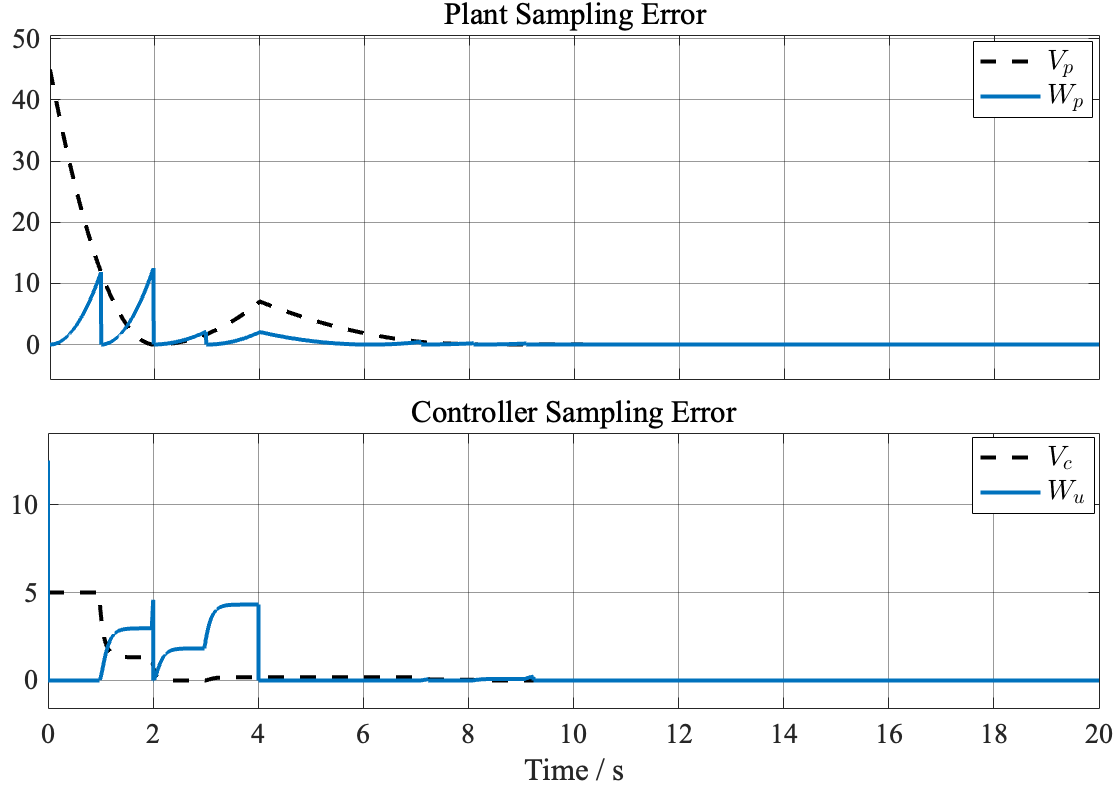}
    \caption{Sampling Error for $k_p=0.5$, $\tau_p = 1$ s, $\tau_c = 2$ s.}
    \label{fig:1-2-0.5 thres}
\end{figure}
\begin{figure}[t]
    \centering
    \includegraphics[width=0.9\columnwidth]{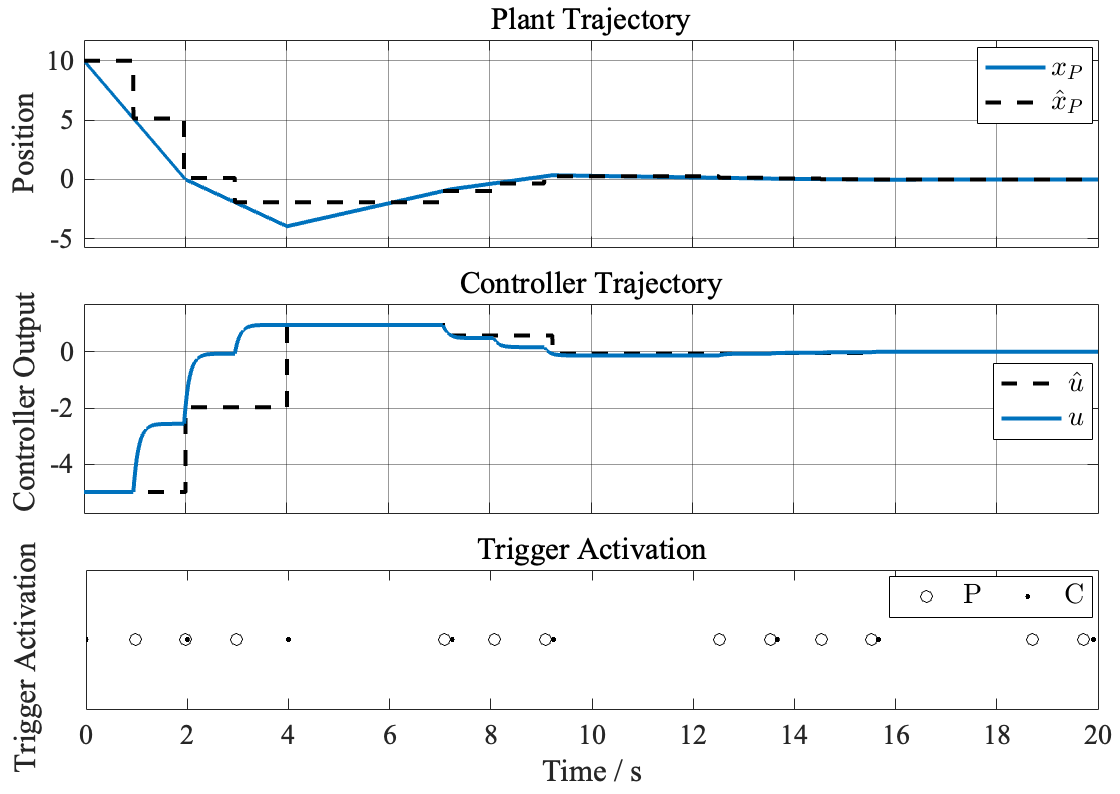}
    \caption{Trajectories for $k_p=0.5$, $\tau_p = 1$ s, $\tau_c = 2$ s.}
    \label{fig:1-2-0.5 traj}
\end{figure}
\begin{figure}[t]
    \centering
    \includegraphics[width=0.9\columnwidth]{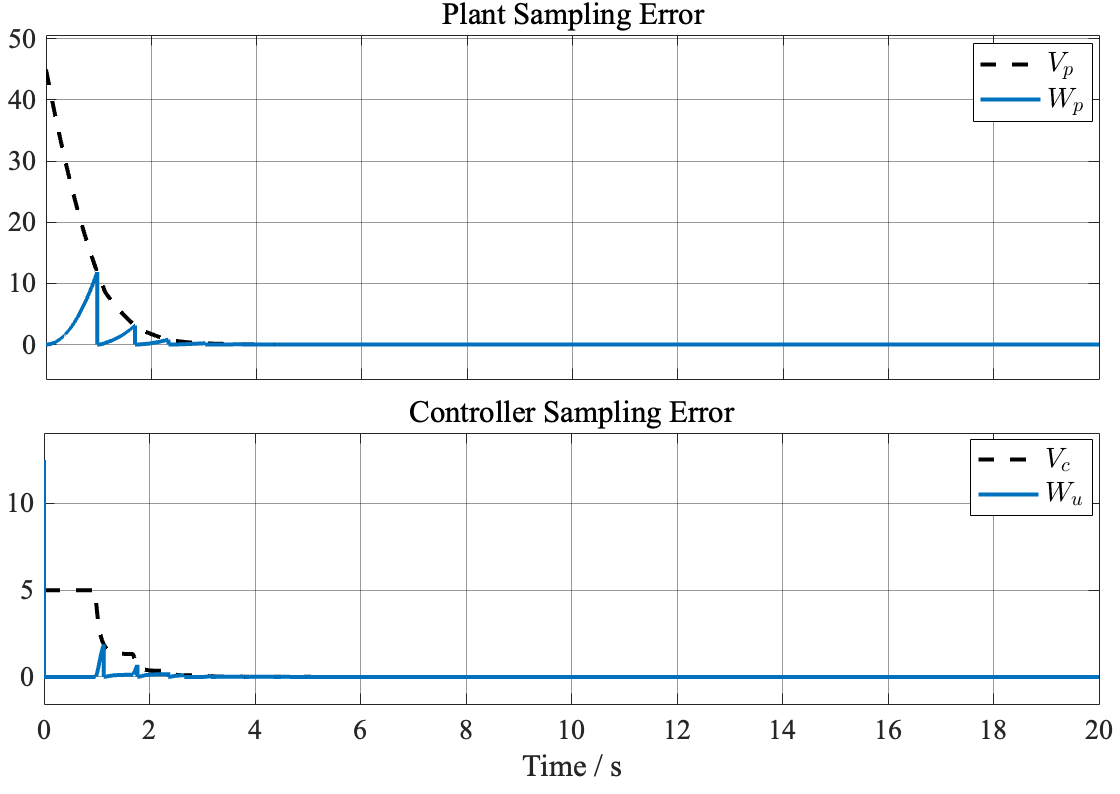}
    \caption{Sampling Error for $k_p=0.5$, $\tau_p = 0.2$ s, $\tau_c = 0.3$ s.}
    \label{fig:0.2-0.3-0.5 thres}
\end{figure}
\begin{figure}[t]
    \centering
    \includegraphics[width=0.9\columnwidth]{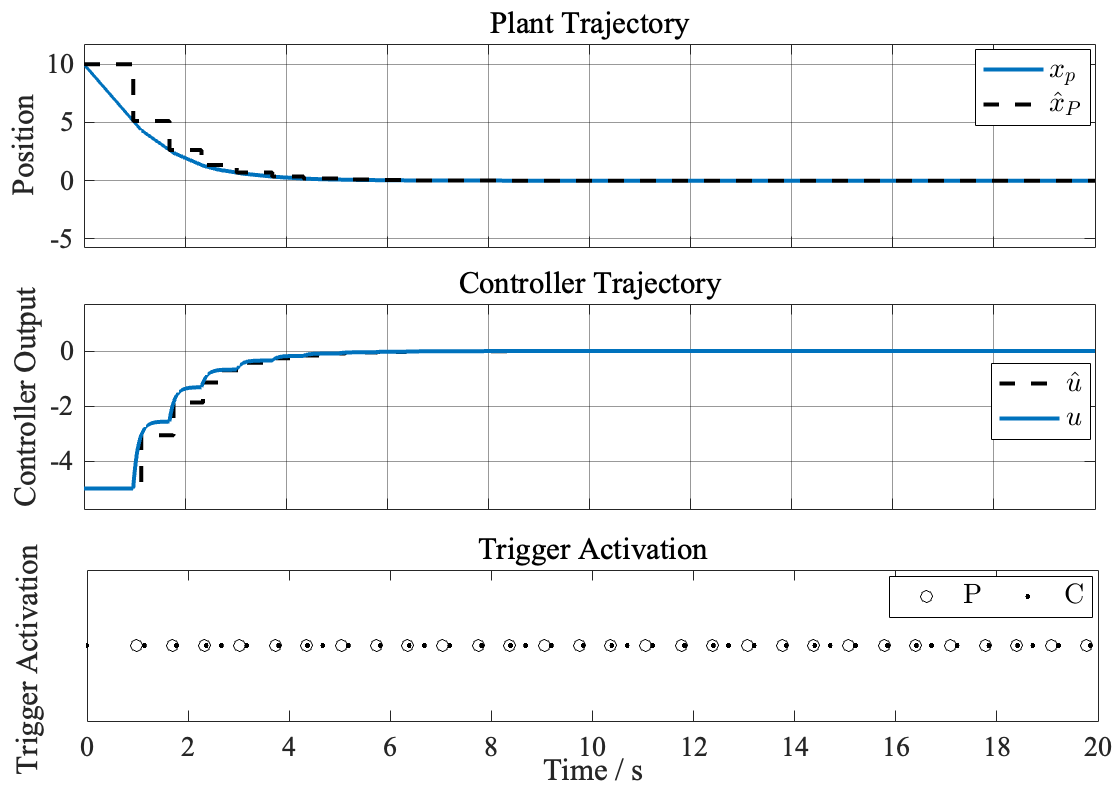}
    \caption{Trajectories for $k_p=0.5$, $\tau_p = 0.2$ s, $\tau_c = 0.3$ s.}
    \label{fig:0.2-0.3-0.5 traj}
\end{figure}
        
To illustrate an asynchronous event-triggered control scheme, we consider a single-integrator plant controlled by a proportional regulator.
The plant dynamics are given by $\dot{x}_p = u+e_u$, with $x_p,u,e_u\in\mathbb{R}$ and $x_p(0,0)= 10$.
The plant state is then conveyed directly to the plant sampler (i.e. $x_m=x_p$), which samples $x_p$ using a zero-order hold scheme with $\tau_P\in \{0.2,1\}$ s and $\tau_{\pi}=60$ s.
We also define the storage function $V_p(q_p) = \frac{1}{2}x_p^2 + \frac{\beta_p}{2} e_p^2$, $\beta_p = 10^{-5}$, and error threshold function
$W_p(q_p) = \frac{1+\beta_p}{2}e_p^2$.

The controller state $x_c\in\mathbb{R}$ tracks $\hat{x}_p$ using $\dot{x}_c = \hat{x}_p - x_c$, with $x_c(0,0)= x_p(0,0)$.
The controller then uses this estimate of the plant state to generate the control input $u = -k_px_c$, where $k_p\in\{0.5,10\}$.
The control input is sampled by a zero-order hold scheme with $\tau_C\in\{0.3,2\}$ s and $\tau_{\kappa}=120$ s.
Similarly, we define the storage function $V_c(q_c) = \frac{1}{10}x_c^2 + \frac{\beta_c}{2} e_u^2$, $\beta_c = 10^{-5}$, and the error threshold function $W_u(q_c) = \frac{1+\beta_c}{2}e_u^2$.

As discussed in Remark \ref{rem:hard}, it is challenging to verify \textit{a priori} whether Assumption \ref{asm:smallgain} holds.
We shall instead illustrate how the underlying principles can inform the refinement of parameter choices.
First we consider the nominal example with $\{k_p=10, \tau_p=1,\tau_c=2\}$.
As depicted in Figure \ref{fig:1-2-10 thres}, both samplers' errors diverge.

One source of instability may be the closed-loop gain magnitude.
Indeed by reducing the plant gain to $k_p=0.5$, the system's errors no longer diverge in Figure \ref{fig:1-2-0.5 thres}. 
We see however in Figure \ref{fig:1-2-0.5 traj} that the plant state overshoots the origin several times, potentially due to insufficiently frequent triggering.

Reducing the minimum sampling intervals to $\tau_p=0.2$ s and $\tau_c=0.3$ s, we observe in Figure \ref{fig:0.2-0.3-0.5 thres} that the error thresholds no longer increase intermittently with plant state overshoot eliminated.
Noting the increase in trigger activations in Figure \ref{fig:0.2-0.3-0.5 traj}, we may thus infer a trade-off between trigger frequency and system overshoot.

\section{Conclusions}\label{sec:conc}
        
We have presented a framework for performing event-triggered control with hybrid nonlinear systems.
Having formulated an additional controller sampler trigger distinct from a system-wide trigger, we have demonstrated that the framework permits more frequent sampling of the controller state while maintaining a minimum time between successive sampling events and asymptotic stability of the origin for the system's state.
Practical considerations regarding parameter selection for the framework have been illustrated with a numerical example.
Future work could consider the implications of channel noise, revising the model to handle exogenous inputs, and integrating online observer schemes into the controller sampler and plant sampler subsystems.

\section*{Acknowledgements}
This paper has been written on the lands of the Boonwurrung and Wurundjeri Woi Wurrung peoples.
This research has been supported by an Australian Government Research Training Program (RTP) Scholarship.
The authors thank Dr Romain Postoyan, Prof Dragan Nešić and Elena Vella for discussions and comments.

\bibliographystyle{IEEEtran}
\bibliography{root}

\end{document}